\documentclass[aps,pra,twocolumn,superscriptaddress,10pt]{revtex4-2}
\usepackage{silence}
\usepackage{amsfonts}
\usepackage{eucal}
\usepackage{amsmath, amsthm, amssymb}
\usepackage{dsfont}

\DeclareFontFamily{OT1}{pzc}{}
\DeclareFontShape{OT1}{pzc}{m}{it}{<-> s * [1.150] pzcmi7t}{}
\DeclareMathAlphabet{\mathpzc}{OT1}{pzc}{m}{it}
\DeclareMathOperator{\tr}{Tr}

\usepackage{hyperref, cleveref}

\usepackage{float}

\theoremstyle{plain}
\newtheorem{theorem}{Theorem}
\newtheorem{lemma}[theorem]{Lemma}

\theoremstyle{definition}

\crefname{equation}{Eq.}{eqns}
\Crefname{equation}{Eq.}{Eqns}

\begin{document}


\title{Designing Stochastic Channels}


\author{Matthew A.\ Graydon}
\email[]{m3graydo@uwaterloo.ca}
\affiliation{Institute for Quantum Computing, University of Waterloo, Waterloo, Ontario N2L 3G1, Canada}
\affiliation{Department of Physics \& Astronomy, University of Waterloo, Waterloo, Ontario N2L 3G1, Canada}

\author{Joshua Skanes-Norman}
\affiliation{Institute for Quantum Computing, University of Waterloo, Waterloo, Ontario N2L 3G1, Canada}
\affiliation{Department of Applied Mathematics, University of Waterloo, Waterloo, Ontario N2L 3G1, Canada}
\affiliation{Keysight Technologies Canada, Kanata, ON K2K 2W5, Canada}

\author{Joel J.\ Wallman}
\affiliation{Institute for Quantum Computing, University of Waterloo, Waterloo, Ontario N2L 3G1, Canada}
\affiliation{Department of Applied Mathematics, University of Waterloo, Waterloo, Ontario N2L 3G1, Canada}
\affiliation{Keysight Technologies Canada, Kanata, ON K2K 2W5, Canada}

\date{\today}

\begin{abstract} 
Stochastic channels are ubiquitous in the field of quantum information because they are simple and easy to analyze. 
In particular, Pauli channels and depolarizing channels are widely studied because they can be efficiently simulated in many relevant quantum circuits.
Despite their wide use, the properties of general stochastic channels have received little attention.
In this paper, we prove that the diamond distance of a general stochastic channel from the identity coincides with its process infidelity to the identity.
We demonstrate with an explicit example that there exist multi-qubit stochastic channels that are not unital.
We then discuss the relationship between unitary 1-designs and stochastic channels. 
We prove that the twirl of an arbitrary quantum channel by a unitary 1-design is always a stochastic channel. 
However, unlike with unitary 2-designs, the twirled channel depends upon the choice of unitary 1-design.
Moreover, we prove by example that there exist stochastic channels that cannot be obtained by twirling a quantum channel by a unitary 1-design.
\end{abstract}

\maketitle


\section{Introduction}

Quantum computers hold tremendous promise, but noise severely limits the reliability of quantum devices in the noisy intermediate-scale quantum (NISQ) era \cite{Preskill2018}. Quantum error correction and noise mitigation methods hinge on assumptions about the noise in question \cite{cao2021,Carignan2019}. A standard assumption is that the noise is a Pauli channel, because such channels are simple and easy to analyze. In this paper, we consider a more general class of stochastic error processes. A stochastic channel is a completely positive and trace-preserving (CPTP) map admitting a set of Hilbert-Schmidt orthogonal Kraus operators containing a strictly positive multiple of the identity \cite{Wallman2015bounding}. 
Operationally, a stochastic channel corresponds to an error occurring with some probability $p\in[0,1)$. Some nice examples are depolarizing channels, decoherence channels, multi-qubit Pauli channels and qudit Weyl-Heisenberg channels \cite{Watrous2018}. In this paper, we identify significant operational properties common to all stochastic channels. 

The diamond distance to the identity \cite{Kitaev1997} (hereafter, the diamond distance) is a metric for quantifying error rates in quantum circuits. The diamond distance can be expressed as a semidefinite program \cite{Watrous2009}, but no efficient protocols have been found to measure it directly for arbitrary channels. 
However, we prove that the diamond distance of a stochastic channel is its process infidelity to the identity (\cref{thm:diamond}); moreover, these measures are simply the probability that an error occurs. The process infidelity to the identity (hereafter, the process infidelity) can be efficiently estimated via randomized benchmarking and generalizations thereof \cite{Emerson2005, Wallman2018, Helsen2020}. Stochastic channels are thus a class of channels wherein the difficult problem of calculating the diamond distance becomes quite easy.

Stochastic channels can be created with low overhead. Twirling over unitary $t$-designs \cite{DiVincenzo2002,Dankert2005,Dankert2009,Cleve2015,Zhu2016,Webb2016,Zhu2017,Haferkamp2020,Skanes2021} is an industry standard technique in quantum information science. A unitary $t$-design is a normalized measure on the unitary group mimicking its $m{^\text{th}}$ Haar moments for all $m\leq t$. The twirl of a quantum channel over a unitary 2-design is a depolarizing channel with the same process infidelity. In this paper, we exhibit a similar mechanism with $t=1$. We point out that channel twirling commutes with the Choi-Jamio{\l}kowski isomorphism (\cref{lem:twirl}). We prove that the twirl of any quantum channel over any unitary $1$-design is a stochastic channel with $p$ the process fidelity of the input channel (\cref{thm: twirl}). We emphasize that noise can be tailored as such without resorting to some potentially higher unitary $t$-designs, \textit{e.g}.\ Clifford groups \cite{Graydon2021}. The twirl of a quantum channel over a unitary 1-design, however, depends on the chosen unitary 1-design. Many unitary 1-designs exist in all finite dimensions \cite{Roy2009}. 

A natural question is whether or not the class of stochastic channels extends beyond the class of mixed unitary channels implemented by unitary operator bases containing the identity. Indeed, the twirl of a quantum channel over a unitary 1-design is unital (\cref{thm: unitalTwirl}). We nevertheless establish a positive answer. We define a noisy rank-2 state preparation channel for a system of two qubits that is stochastic but not unital. We also prove that all stochastic qubit channels are unital (\cref{thm: qubits}).

The balance of this paper is structured as follows. In \cref{sec:preliminaries}, we recall preliminary notions from quantum information theory in order to establish three equivalent definitions of stochastic channels. We then define via \cref{examplePhi} a stochastic channel that is not unital. In \cref{sec:diamonds}, we prove that the diamond distance and process infidelity are equal for arbitrary stochastic channels. In \cref{sec:twirling}, we discuss connections between stochastic channels and unitary 1-designs and we prove \cref{lem:twirl}, \cref{thm: twirl}, and \cref{thm: unitalTwirl}. Finally, in \cref{sec: conclusion} we conclude and discuss directions for future research.

\section{Stochastic Channels}\label{sec:preliminaries}

We begin by defining stochastic channels, which are informally a noiseless process probabilistically mixed with orthogonal errors.
A formal definition of a stochastic channel (in particular, the notion of orthogonal channels) is stated most naturally in terms of the Choi state, which we now review.
Let $\mathpzc{H}$ be a finite-dimensional complex Hilbert space and $\mathfrak{B}(\mathpzc{H})$ be the set of bounded complex-linear maps from $\mathpzc{H}$ to itself equipped with the Hilbert-Schmidt inner product,
\begin{align}
    \langle\!\langle A | B \rangle\!\rangle \equiv \tr A^\dagger B \quad \forall A, B \in \mathfrak{B}(\mathpzc{H}).
\end{align}
Then $\mathfrak{B}(\mathpzc{H})$ is itself a finite-dimensional complex Hilbert space and so $\mathfrak{B}(\mathfrak{B}(\mathpzc{H}))$ is well-defined.
Canonically, $\mathpzc{H}$ is the set of complex vectors, that is, $\mathpzc{H} = \mathbb{C}^d$ for some finite positive integer $d$, and $\mathfrak{B}(\mathpzc{H})$ are $d \times d$ complex matrices.
We can represent the elements of $\mathfrak{B}(\mathbb{C}^d)$ as vectors using the column-stacking vectorization map $\text{col}$, and then elements $\Phi$ of $\mathfrak{B}(\mathfrak{B}(\mathbb{C}^d))$ by their \textit{Choi-Jamio{\l}kowski state} (hereafter their Choi state),
\begin{equation}
\mathpzc{J}(\Phi)\equiv \frac{1}{d}\Phi\otimes\text{id}_{d}\left(\text{col}(\mathds{1}_{d})\text{col}(\mathds{1}_{d})^{\dagger}\right),
\label{CJstate}
\end{equation} 
where $\mathds{1}_d \in \mathfrak{B}(\mathbb{C}^d)$ is the identity operator and $\text{id}_d$ is the identity element of $\mathfrak{B}(\mathfrak{B}(\mathbb{C}^d))$.
Applying a singular value decomposition to $\mathpzc{J}(\Phi)$ and using the fact that the vectorization map is invertible, we find that there exist suitable vectors $l_{k},r_{k}\in\mathbb{C}^{d^{2}}$ and matrices $L_{k},R_{k}\in\mathfrak{B}(\mathbb{C}^{d})$ so that
\begin{align}\label{specJ}
    \mathpzc{J}(\Phi) = \frac{1}{d}\sum_{k = 1}^{d^2} l_k r_k^\dagger \equiv \frac{1}{d}\sum_{k = 1}^{d^2} \text{col}(L_k^\text{T})\text{col}(R_k^\text{T})^\dagger\text{.}
\end{align}
Using the vectorization identity
\begin{equation}
\text{col}(ABC)=(C^{\text{T}}\!\otimes\!A)\text{col}(B)  \quad \forall A,B,C\in\mathfrak{B}(\mathbb{C}^d),
\label{vecId}
\end{equation}
we thus see that we can write $\Phi$ in a (non-unique) Kraus form as
\begin{align}\label{eq:KrausHellwig}
    \Phi(X) = \sum_{k = 1}^{d^2} L_k X R_k^\dagger.
\end{align}
Using the vectorization identity
\begin{equation}
\tr A^\dagger B = \text{col}(A)^\dagger \text{col}(B),
\label{vecTraceID}
\end{equation}
the orthogonality of the vectors in the singular value decomposition implies that for all $j \neq k$, we have
\begin{align}
    \tr L_j^\dagger L_k = \tr R_j^\dagger R_k = 0.
\end{align}
For any CPTP map, the Choi state is a positive matrix, so that $L_k = R_k$ for all $k$ and we can write
\begin{align}
    \Phi = \sum_{k=1}^{d^2} \text{Ad}_{L_k}
\end{align}
where for any matrix $A\in\mathfrak{B}(\mathbb{C}^{d})$ we define
\begin{equation}\label{AdDef}
\text{Ad}_{A}:\mathfrak{B}(\mathbb{C}^{d})\longrightarrow \mathfrak{B}(\mathbb{C}^{d})::x\longmapsto AxA^{\dagger}\text{.}
\end{equation}
Using the cyclic invariance of the trace, the trace preservation condition on a CPTP map is
\begin{equation}\label{eq:TP}
\sum_{k=1}^{d^2} L_k^\dagger L_k =\mathds{1}_{d}\text{.}
\end{equation}
A \textit{stochastic channel} is a CPTP map $\Phi\in\mathfrak{B}(\mathfrak{B}(\mathbb{C}^{d}))$ such that for some $0<\lambda\leq 1$
\begin{equation}
\mathpzc{J}(\Phi)\text{col}(\mathds{1}_{d})=\lambda\text{col}(\mathds{1}_{d})\text{.}
\end{equation}
Put otherwise, a stochastic channel is a CPTP map whose set of canonical Kraus operators contains a strictly positive multiple of the identity. A stochastic channel is therefore a map 
\begin{equation}\label{eq:StochasticDefinition}
\Phi:\mathfrak{B}(\mathbb{C}^{d})\longrightarrow\mathfrak{B}(\mathbb{C}^{d})::x\longmapsto \lambda x+(1-\lambda)\Phi^{\perp}(x)\text{,}
\end{equation}
where $\Phi^{\perp}\in\mathfrak{B}(\mathfrak{B}(\mathbb{C}^{d}))$ is an error channel that is orthogonal to the identity, that is
\begin{equation}
\mathpzc{J}(\Phi^{\perp})\text{col}(\mathds{1}_{d})=0\text{.}
\end{equation}
The set of stochastic channels is manifestly convex via \cref{eq:StochasticDefinition}.
The class of stochastic channels contains channels of known types. A mixed unitary channel is a quantum channel such that its Kraus decomposition may be taken in terms of a set of appropriately subnormalized unitary matrices. A mixed unitary channel is therefore always unital, which is to say that it leaves the maximally mixed state invariant. Indeed, some quantum channels might not appear to be mixed unitary at first sight. It turns out that question is actually NP-Hard \cite{Lee2020}. We proceed now by way of an example that our definition of stochastic channels encompasses channels outside the mixed unitary class. Define the following $4\times 4$ matrices with $\lambda\in(0,1]$
\begin{align}
A_{1}&\equiv\sqrt{\lambda}\mathds{1}_{2}\otimes\mathds{1}_{2}\text{,}\label{Ae1}\\
A_{2}&\equiv\sqrt{1-\lambda}\sigma_{z}\otimes|0\rangle\!\langle 0|\text{,}\label{Ae2}\\
A_{3}&\equiv\sqrt{1-\lambda}\sigma_{z}\otimes|0\rangle\!\langle 1|\text{.}\label{Ae3}
\end{align}
where $\sigma_{z}$ is the usual Pauli Z matrix. Defining the quantum channel
\begin{align}
\Phi = \sum_{k=1}^3 \text{Ad}_{A_k},
\label{examplePhi}
\end{align}
one can easily verify that $\Phi$ is a stochastic quantum channel. One can check, furthermore, that $\Phi(\mathds{1}_{4})=\lambda\mathds{1}_{4}+(1-\lambda)\mathds{1}_{2}\otimes2|0\rangle\langle0|\neq\mathds{1}_{4}$. $\Phi$ prepares a rank-2 state with probability $(1-\lambda)$, and with probability $\lambda$ the stochastic channel $\Phi$ is noiseless. 
This construction generalizes to all even dimensions larger than 2.
In \cref{app: qubits} we prove that all stochastic qubit channels are unital. 
We leave open the question of whether or not all stochastic channels are unital for $d = 3$.

\section{Faithful Diamonds}\label{sec:diamonds}
In this section, we prove that stochastic channels are such that their diamond distance equals their process infidelity. We recall that the \textit{trace norm} is $\|\cdot\|_{1}:\mathfrak{B}(\mathbb{C}^{d})\longrightarrow\mathbb{R}_{\geq 0}::X\longmapsto\mathrm{Tr}\sqrt{X^{\dagger}X}$. The \textit{induced trace norm} on superoperators is defined via
\begin{eqnarray}
\|\cdot\|_{1}&:&\mathfrak{B}(\mathfrak{B}(\mathbb{C}^{d}))\longrightarrow\mathbb{R}_{\geq 0}\nonumber\\
&::&\Gamma\longmapsto \text{max}\left\{\|\Phi(X)\|_{1}\;\big|\;\|X\|_{1}\leq 1\right\}\text{.}
\end{eqnarray}
The completely bounded trace norm \cite{Watrous2018} is also commonly referred to as the \textit{diamond norm} \cite{Kitaev1997} and defined as follows
\begin{eqnarray}
\|\cdot\|_{\diamond}&:&\mathfrak{B}\big(\mathfrak{B}(\mathbb{C}^{d})\big)\longrightarrow\mathbb{R}_{\geq 0}\nonumber\\
&::&\Gamma\longmapsto \|\Gamma\otimes\text{id}_{d}\|_{1}\text{.}
\label{dDef}
\end{eqnarray}
The diamond norm induces for any quantum channel $\Phi$ its \textit{diamond distance from the identity}
\begin{equation}
r_{\diamond}(\Phi)\equiv\frac{1}{2}\|\Phi-\text{id}_{d}\|_{\diamond}\text{.}
\end{equation}
The \textit{process fidelity to the identity} of $\Phi$ is
\begin{align}\label{ProcessFidelity}
F_e(\Phi) &= \tr \left(\mathpzc{J}(\text{id}_d)^\dagger \mathpzc{J}(\Phi)\right) \\
&= \tr \left(\text{col}(\mathds{1}_d) \text{col}(\mathds{1}_d)^\dagger \mathpzc{J}(\Phi)\right) \\
&= \text{col}(\mathds{1}_d)^\dagger \mathpzc{J}(\Phi) \text{col}(\mathds{1}_d). 
\end{align}
The \textit{process infidelity} is \cite{Fuchs1999} 
\begin{equation}
r(\Phi)\equiv 1-F_{e}(\Phi)\text{.}
\end{equation}

The following theorem generalizes a similar result for generalized Pauli channels \cite{Magesan2012}. Note that, in our more general case of interest, the canonical Kraus operators for an arbitrary stochastic channel need not be unitary.
\begin{theorem}\label{thm:diamond}\textit{Let} $\Phi\in\mathfrak{B}(\mathfrak{B}(\mathbb{C}^{d}))$ \textit{be a stochastic quantum channel. Then}
\begin{equation}
r_{\diamond}(\Phi)=r(\Phi)\text{.}
\end{equation}
\end{theorem}
\begin{proof} 
For any CPTP map, we have the bound~\cite{Wallman2014}
\begin{align}
    r_\diamond(\Phi) \geq r(\Phi).
\end{align}
Therefore we need only prove that for a stochastic map $\Phi$, 
\begin{align}
    r_\diamond(\Phi) \leq r(\Phi).
\end{align}
To prove this, recall that by the definition of a stochastic channel we have 
\begin{align}
    \Phi = \sum_k \text{Ad}_{B_k}
\end{align}
with $B_1 = \sqrt{\lambda} \mathds{1}_d$ for some $\lambda \in (0, 1)$.
By the triangle inequality, we therefore have
\begin{align}
    \|\Phi-\text{id}_{d}\|_\diamond 
    &\leq (1 - \lambda) \|\text{id}_{d}\|_\diamond + \|\sum_{k = 2}^{d^2}  \text{Ad}_{B_k}\|_\diamond \notag\\
    &\leq 1 - \lambda + \|\sum_{k = 2}^{d^2}  \text{Ad}_{B_k}\|_\diamond
\end{align}
where we have used the fact that $\| \text{id}_{d} \|_\diamond = 1$.
Now from the definition of the diamond norm and the triangle inequality, we have
\begin{align}\label{eq:orthNorm}
    \| \sum_{k = 2}^{d^2}  \text{Ad}_{B_k} \|_\diamond 
    &= \sup_{\rho} \| \sum_{k = 2}^{d^2} \text{Ad}_{B_k\otimes \mathds{1}_d}(\rho) \|_1 \notag\\
    &= \sup_{\rho} \sum_{k = 2}^{d^2} \| \text{Ad}_{B_k\otimes \mathds{1}_d}(\rho) \|_1.
\end{align}
For any operator $B$ and any positive operator $\rho$, $B \rho B^\dagger$ is positive and so $\|\text{Ad}_B(\rho) \|_1 = \tr B \rho B^\dagger = \tr B^\dagger B \rho$ by the cyclic invariance of the trace.
Therefore using \cref{eq:TP}, we find that for any CPTP map $\Phi$,
\begin{align}\label{eq:orthNorm2}
    \| \sum_{k = 2}^{d^2}  \text{Ad}_{B_k} \|_\diamond 
    &\leq \sup_{\rho} \sum_{k = 2}^{d^2} \tr (B_k^\dagger B_k \otimes \mathds{1}_d) \rho \notag\\
    &= \sup_{\rho} \tr \big((\mathds{1}_{d} - B_1^\dagger B_1) \otimes \mathds{1}_{d}\big) \rho.
\end{align}
For a stochastic channel, $B_1 = \sqrt{\lambda} \mathds{1}_d$ and so
\begin{align}\label{eq:orthNorm3}
    \| \sum_{k = 2}^{d^2}  \text{Ad}_{B_k} \|_\diamond 
    &\leq 1 - \lambda.
\end{align}
We thus have
\begin{align}
r_\diamond(\Phi) = 1 - \lambda.
\end{align}
\end{proof}

\section{Twirling}\label{sec:twirling}

We now show that unitary 1-designs can be used to construct stochastic channels.
We also show that there exist stochastic channels that cannot be constructed using a unitary 1-design.
Let $\mathsf{U}(\mathpzc{H})$ be the unitary group on a Hilbert space $\mathpzc{H}$. 
Let $\mu:\Sigma(\mathsf{U}(\mathbb{C}^d))\longrightarrow\mathbb{R}_{\geq 0}$ be a normalized measure on the unitary group $\mathsf{U}(\mathbb{C}^d)$, where $\Sigma(\mathsf{U}(\mathbb{C}^d))$ denotes the Borel $\sigma$-algebra of $\mathsf{U}(\mathbb{C}^d)$ \cite{halmos2013measure}. 
Then the $\mu$-\textit{twirl} of a map $\Phi\in\mathfrak{B}(\mathfrak{B}(\mathbb{C}^{d}))$ is the map
\begin{equation}
\Phi_{\mu}=\int_{\mathsf{U}(\mathbb{C}^d)}d\mu(u)\mathrm{Ad}_{u}\circ\Phi\circ\mathrm{Ad}_{u^{\dagger}}.
\label{twirlDef}
\end{equation}

\begin{lemma}\label{lem:twirl}
Let $\Phi \in \mathfrak{B}(\mathfrak{B}(\mathbb{C}^d))$. Let $\mu$ be a measure on the unitary group. Then
\begin{equation}
\mathpzc{J}\left(\Phi_{\mu}\right)=\int_{\mathsf{U}(\mathbb{C}^d)}d\mu(u)\mathrm{Ad}_{u \otimes \bar{u}}\left(\mathpzc{J}(\Phi)\right).
\label{twirledJ}
\end{equation} 
\end{lemma}

\begin{proof}
Let $L_k$ and $R_k$ satisfy \cref{eq:KrausHellwig}.
For fixed elements $u,v \in \mathsf{U}(\mathbb{C}^d)$,
\begin{align}
    \Phi' = \text{Ad}_u \circ \Phi \circ \text{Ad}_{v}
\end{align}
can be written in the form
\begin{align}\label{eq:KrausHellwigPrime}
    \Phi'(X) = \sum_{k = 1}^{d^2} u L_k v X v^\dagger R_k^\dagger u^\dagger.
\end{align} 
By \cref{specJ,vecId}, we thus have
\begin{align}\label{eq:compositeUV}
    \mathpzc{J}(\Phi') 
    &= \frac{1}{d}\sum_{k = 1}^{d^2} \text{col}(v^\text{T} L_k^\text{T} u^\text{T}) \text{col}(\bar{u} R_k^\text{T} \bar{v})^\dagger \notag\\
    &= \frac{1}{d}u \otimes v^{\text{T}} \left( \sum_{k = 1}^{d^2} \text{col}(L_k^\text{T}) \text{col}(R_k^\text{T})^\dagger \right) (v^\dagger \otimes \bar{u})^\dagger \notag\\
    &= (u \otimes v^{\text{T}}) \mathpzc{J}(\Phi') (v \otimes u^{\text{T}}).
\end{align}
Setting $v = u^\dagger$ and averaging over $\mu$ completes the proof.
\end{proof}

There are many equivalent definitions of unitary 1-designs \cite{Dankert2005}\cite{Haferkamp2020}. A \textit{unitary} $1$\textit{-design} (in dimension $d$) is a normalized measure $\mu$ on the unitary group such that
\begin{equation} 
\int_{\mathsf{U}(\mathbb{C}^{d})}d\mu(u)\text{Ad}_u = \int_{\mathsf{U}(\mathbb{C}^{d})}d\mu_{\text{H}}(u) \text{Ad}_u\text{,}
\label{1prop}
\end{equation}
where $\mu_{\text{H}}$ is the normalized Haar measure on $\mathsf{U}(\mathbb{C}^{d})$ \cite{Knapp2016}. We remark that 
\begin{align}
\int_{\mathsf{U}(\mathbb{C}^{d})}d\mu(u)uXu^{\dagger}=\frac{\mathrm{Tr}(X)}{d}\mathds{1}_{d}\text{,}\label{schur}\\
\int_{\mathsf{U}(\mathbb{C}^{d})}d\mu(u) u \otimes \bar{u} =\frac{1}{d}\mathrm{col}(\mathds{1}_{d}) \mathrm{col}(\mathds{1}_{d})^{\dagger}\text{,}
\label{schurNat}
\end{align} 
are then consequences of Schur's lemma \cite{Kirillov1976} where $X \in \mathfrak{B}(\mathbb{C}^d)$. 

The close connection between unitary 1-designs and irreducible unitary representations (in dimension $d$) \cite{Roy2009} means that many unitary 1-designs exist. We especially recommend Bannai \textit{et al}.\ \cite{Bannai2019} for explicit examples. The GAP system \cite{GAP4} could also be harvested. The Pauli group on a single qubit is a unitary 1-design. The multipartite Weyl-Heisenberg group \cite{Appleby2005} is a unitary 1-design for any collection of quantum systems. Another example is always the Clifford group \cite{Graydon2021}. 

We now prove that the twirl of a quantum channel over a unitary 1-design is a stochastic quantum channel. Our proof reveals that the resulting stochastic quantum channel is noiseless with probability the process fidelity of the input channel, which underlies protocols that are designed to estimate the process fidelity~\cite{Erhard2019} or suppress coherent errors~\cite{Wallman2016}.

\begin{theorem}\label{thm: twirl}
\textit{Let} $\Phi\in\mathfrak{B}(\mathfrak{B}(\mathbb{C}^{d}))$ \textit{be quantum channel with nonzero process fidelity.}  \textit{ Let} $\mu$ \textit{be a unitary $1$-design. Then} $\Phi_{\mu}$ \textit{is a stochastic channel.}
\end{theorem}
\begin{proof} We shall prove that $\text{col}(\mathds{1}_{d})$ is an eigenvector of $\mathpzc{J}(\Phi_{\mu})$ with a strictly positive eigenvalue, hence the identifications in \cref{sec:preliminaries} yield the result. We calculate
\begin{align}
&\hspace{0.15cm}\mathpzc{J}(\Phi_{\mu})\text{col}(\mathds{1}_{d})\notag\\
&=\int_{\mathsf{U}(\mathbb{C}^{d})}d\mu(u) u \otimes \bar{u} \mathpzc{J}(\Phi) u^\dagger \otimes u^{\text{T}} \text{col}(\mathds{1}_{d}) \tag{by \cref{lem:twirl}}\\
&=\int_{\mathsf{U}(\mathbb{C}^{d})}d\mu(u) u \otimes \bar{u} \mathpzc{J}(\Phi)\text{col}(\mathds{1}_{d}) \tag{by \cref{vecTraceID}}\\
&=\frac{1}{d} \mathrm{col}(\mathds{1}_{d})\mathrm{col}(\mathds{1}_{d})^{\dagger}\mathpzc{J}(\Phi)\text{col}(\mathds{1}_{d}) \tag{by \cref{schurNat}}\\[0.2cm]
&=F_{e}(\Phi)\text{col}(\mathds{1}_{d})\notag\text{,}
\end{align}
where the final equality follows from \cref{ProcessFidelity}.
\end{proof}

The example defined via Eqs.~\eqref{Ae1}-\eqref{examplePhi} is not a unital stochastic channel. In light of the foregoing example and the following theorem, we conclude that some stochastic channels cannot be accessed via the twirl of any unitary 1-design. 
\begin{theorem}\label{thm: unitalTwirl}
\textit{Let} $\Phi\in\mathfrak{B}(\mathfrak{B}(\mathbb{C}^{d}))$ \textit{be quantum channel. Let} $\mu$ \textit{be a unitary $1$-design. Then} $\Phi_{\mu}$ \textit{is unital.}
\end{theorem}
\begin{proof}
We calculate
\begin{align}
\Phi_{\mu}(\mathds{1}_{d})&=\int_{\mathsf{U}(\mathbb{C}^{d})}d\mu(u)\text{Ad}_{u}\circ\Phi\circ\text{Ad}_{u^{\dagger}}(\mathds{1}_{d})\label{explicitStep1}\tag{by \cref{twirlDef}}\\
&=\int_{\mathsf{U}(\mathbb{C}^{d})}d\mu(u)\text{Ad}_{u}(\Phi(\mathds{1}_{d}))\label{explicitStep2}\tag{by \cref{AdDef}}\\
&=\frac{\tr \Phi(\mathds{1}_d)}{d}\mathds{1}_{d}\label{explicitStep4}\tag{by \cref{schur}}\\[0.3cm]
&=\mathds{1}_{d}\label{explicitStep5}\notag\text{,}
\end{align}
where the final step follows as $\Phi$ is trace preserving.
\end{proof}

\section{Conclusion}\label{sec: conclusion}
We have considered properties of the class of stochastic channels. We proved that the diamond distance of a stochastic channel is its process infidelity (\cref{thm:diamond}). We illuminated the relationship between stochastic channels and unitary 1-designs. In general, we pointed out that channel twirling commutes with the Choi-Jamio{\l}kowski isomorphism (\cref{lem:twirl}). We then demonstrated a way to access stochastic channels. The twirl of a quantum channel over a unitary 1-design is a stochastic channel (\cref{thm: twirl}). Next, we proved that twirling over unitary 1-designs results in unital channels (\cref{thm: unitalTwirl}.) We provided an example of a stochastic channel that is not unital. In \cref{app: qubits}, we prove that all qubit stochastic channels are unital (\cref{thm: qubits}). 

There are some potential avenues for future research into designing stochastic channels. One open question is whether or not a quantum channel $\Phi$ exists outside the class of stochastic channels but with $r_{\diamond}(\Phi)=r(\Phi)$. Another open question is whether or not every qutrit stochastic channel is a mixed unitary channel. Future work might also explore effects obtained via twirling over different unitary 1-designs. We observe from our proof of \cref{thm: twirl} that twirling a quantum channel over any unitary 1-design results in a  stochastic channel whose canonical Kraus operators can depend on the unitary 1-design in question. For instance, consider a group unitary operator basis $\mathsf{B}$ and then its conjugation via any element of the unitary group, $\text{Ad}_{U}(\mathsf{B})$. The twirl with respect to the latter is a rotated version of the former.
The difference between the two twirled channels is subtle and can be missed, e.g., in \cite{Wallman2016}, however, since both channels have the same diamond distance and process fidelity, the difference is typically unimportant.

Finally, the set of stochastic channels is convex. 
A natural avenue for future research is to determine the shape of this convex set.

\appendix

\section{Stochastic Qubit Channels}\label{app: qubits}
In this appendix we prove that all stochastic qubit channels are unital. 
\begin{theorem}\label{thm: qubits}
Let $\Xi\in\mathfrak{B}(\mathfrak{B}(\mathbb{C}^{2}))$ be a stochastic quantum channel. Then $\Xi$ is unital.
\end{theorem}
\begin{proof}
Any completely positive and trace preserving map $\Xi\in\mathfrak{B}(\mathfrak{B}(\mathbb{C}^{2}))$  can be decomposed as $\Xi=\text{Ad}_{u}\circ\Phi\circ\text{Ad}_{v}$ with $u, v\in\mathsf{SU}(2)$ and where $\mathpzc{J}(\Phi)$ is the following matrix 
\begin{equation}\label{BZchoi}\hspace{-0.15cm} \begin{small}{\begin{pmatrix}
    1 + \eta_{z}+\kappa_{z} & 0 & \kappa_{x}+i\kappa_{y} & \eta_{x}+\eta_{y}\\
    0 & 1 -\eta_{z}+\kappa_{z} & \eta_{x}-\eta_{y} &\kappa_{x}+i\kappa_{y}\\
    \kappa_{x}-i\kappa_{y} &\eta_{x}-\eta_{y} & 1 -\eta_{z}-\kappa_{z}&0\\
    \eta_{x}+\eta_{y} & \kappa_{x}-i\kappa_{y} & 0 & 1 + \eta_{z}-\kappa_{z}
    \end{pmatrix}}\end{small}
\end{equation}
with parameters $\eta_{x},\eta_{y},\eta_{z}$ characterizing the distortion of the Bloch ball, and $\kappa_{x},\kappa_{y},\kappa_{z}$ defining the Bloch sphere coordinates of the action of $\Phi$ on $\mathds{1}_{2}/2$ \cite{Bengtsson2017}. In particular, $\kappa_{x}=\kappa_{y}=\kappa_{z}=0$ if and only if $\Phi$ is unital. There are constraints on the foregoing parameters in order to ensure completely positivity and trace preservation. We shall, however, need only the facts $\kappa_{x},\kappa_{y},\kappa_{z}\in\mathbb{R}$. 
By \cref{eq:compositeUV}, the Choi state of $\Xi$ is
\begin{align}
\mathpzc{J}(\Xi) = (u \otimes v^{\text{T}}) \mathpzc{J}(\Phi) (v \otimes u^{\text{T}}) \text{.}    
\end{align}
If $\Xi$ is stochastic, then there exists some $\lambda > 0$ such that
\begin{align}
    \mathpzc{J}(\Xi)\text{col}(\mathds{1}_{2})=\text{col}(\mathds{1}_{2})\lambda,
\end{align}
or, multiplying both sides by $(u \otimes v^{\text{T}})^\dagger$ and using \cref{vecId} gives
\begin{align}
    \mathpzc{J}(\Phi)\text{col}((vu)^{\text{T}}) = \lambda\text{col}(\overline{vu}).
\end{align}
$\mathsf{SU}(2)$ is a group, so we can write $(vu)^{\text{T}}$ as
\begin{equation}
        (vu)^{\text{T}} =\begin{small}{\begin{pmatrix}
        \hspace{0.23cm}\gamma_{1} &\gamma_{2}\\-\overline{\gamma_{2}} &\overline{\gamma_{1}}
        \end{pmatrix}}\end{small}\text{,}
\end{equation}
with $\gamma_{1},\gamma_{2}\in\mathbb{C}$ such that $|\gamma_{1}|^{2}+|\gamma_{2}|^{2}=1$.
We then have
\begin{align}\label{eq:qubitCase}
    \mathpzc{J}(\Phi) \begin{pmatrix} \gamma_1  \\ -\overline{\gamma_2} \\ \gamma_2 \\ \overline{\gamma_1} \end{pmatrix} = \lambda \begin{pmatrix} \overline{\gamma_1} \\ \overline{\gamma_2} \\ -\gamma_2 \\ \gamma_1 \end{pmatrix}.
\end{align}
Subtracting the conjugate of the fourth equation and using the fact that the parameters of $\mathpzc{J}(\Phi)$ are real gives
\begin{align}\label{eq:Gamma1}
    -\kappa_z \gamma_1 = (\kappa_x + i \kappa_y) \gamma_2.
\end{align}
Similarly, adding the conjugate of the second equation to the third equation and using the fact that the parameters of $\mathpzc{J}(\Phi)$ are real gives
\begin{align}\label{eq:Gamma2}
    \kappa_z \gamma_2 =  (\kappa_x - i \kappa_y)
    \gamma_1.
\end{align}
If $\gamma_{1}=0$ then \cref{eq:Gamma1} yields $\kappa_{x}=-i\kappa_{y}$, whence by the reality of $\kappa_{x}$ and $\kappa_{y}$ we have $\kappa_x = \kappa_y = 0$ and then by \cref{eq:Gamma2} we conclude that $\kappa_{z}=0$. 
If instead $\gamma_{2}=0$ then we reach the same conclusion. If both $\gamma_{1}$ and $\gamma_{2}$ are nonzero then we multiply \cref{eq:Gamma1} and \cref{eq:Gamma2} to find
\begin{align}
    -\kappa_z^2 \gamma_1 \gamma_2 &= (\kappa_x - i \kappa_y)
    \gamma_1(\kappa_x + i \kappa_y) \gamma_2 \notag\\
    &= (\kappa_x^2 + \kappa_y^2) \gamma_1 \gamma_2,
\end{align}
and so we must have $\kappa_x = \kappa_y = \kappa_z = 0$ and we conclude again that $\Xi$ is unital.
\end{proof}
\begin{acknowledgments}
We would like to thank Arnaud Carignan-Dugas and Raymond Laflamme for some illuminating discussions. MG acknowledges support from NSERC DG2018-03968. This research was undertaken thanks in part to funding from the Canada First Research Excellence Fund, the Government of Ontario, and the Government of Canada through NSERC.
\end{acknowledgments}

\bibliography{GSNW}

\begin{thebibliography}{34}%
\makeatletter
\providecommand \@ifxundefined [1]{%
 \@ifx{#1\undefined}
}%
\providecommand \@ifnum [1]{%
 \ifnum #1\expandafter \@firstoftwo
 \else \expandafter \@secondoftwo
 \fi
}%
\providecommand \@ifx [1]{%
 \ifx #1\expandafter \@firstoftwo
 \else \expandafter \@secondoftwo
 \fi
}%
\providecommand \natexlab [1]{#1}%
\providecommand \enquote  [1]{``#1''}%
\providecommand \bibnamefont  [1]{#1}%
\providecommand \bibfnamefont [1]{#1}%
\providecommand \citenamefont [1]{#1}%
\providecommand \href@noop [0]{\@secondoftwo}%
\providecommand \href [0]{\begingroup \@sanitize@url \@href}%
\providecommand \@href[1]{\@@startlink{#1}\@@href}%
\providecommand \@@href[1]{\endgroup#1\@@endlink}%
\providecommand \@sanitize@url [0]{\catcode `\\12\catcode `\$12\catcode
  `\&12\catcode `\#12\catcode `\^12\catcode `\_12\catcode `\%12\relax}%
\providecommand \@@startlink[1]{}%
\providecommand \@@endlink[0]{}%
\providecommand \url  [0]{\begingroup\@sanitize@url \@url }%
\providecommand \@url [1]{\endgroup\@href {#1}{\urlprefix }}%
\providecommand \urlprefix  [0]{URL }%
\providecommand \Eprint [0]{\href }%
\providecommand \doibase [0]{https://doi.org/}%
\providecommand \selectlanguage [0]{\@gobble}%
\providecommand \bibinfo  [0]{\@secondoftwo}%
\providecommand \bibfield  [0]{\@secondoftwo}%
\providecommand \translation [1]{[#1]}%
\providecommand \BibitemOpen [0]{}%
\providecommand \bibitemStop [0]{}%
\providecommand \bibitemNoStop [0]{.\EOS\space}%
\providecommand \EOS [0]{\spacefactor3000\relax}%
\providecommand \BibitemShut  [1]{\csname bibitem#1\endcsname}%
\let\auto@bib@innerbib\@empty
\bibitem [{\citenamefont {Preskill}(2018)}]{Preskill2018}%
  \BibitemOpen
  \bibfield  {author} {\bibinfo {author} {\bibfnamefont {J.}~\bibnamefont
  {Preskill}},\ }\href@noop {} {\bibfield  {journal} {\bibinfo  {journal}
  {Quantum}\ }\textbf {\bibinfo {volume} {2}},\ \bibinfo {pages} {79} (\bibinfo
  {year} {2018})}\BibitemShut {NoStop}%
\bibitem [{\citenamefont {Cao}\ \emph {et~al.}(2021)\citenamefont {Cao},
  \citenamefont {Lin}, \citenamefont {Kribs}, \citenamefont {Poon},
  \citenamefont {Zeng},\ and\ \citenamefont {Laflamme}}]{cao2021}%
  \BibitemOpen
  \bibfield  {author} {\bibinfo {author} {\bibfnamefont {N.}~\bibnamefont
  {Cao}}, \bibinfo {author} {\bibfnamefont {J.}~\bibnamefont {Lin}}, \bibinfo
  {author} {\bibfnamefont {D.}~\bibnamefont {Kribs}}, \bibinfo {author}
  {\bibfnamefont {Y.-T.}\ \bibnamefont {Poon}}, \bibinfo {author}
  {\bibfnamefont {B.}~\bibnamefont {Zeng}},\ and\ \bibinfo {author}
  {\bibfnamefont {R.}~\bibnamefont {Laflamme}},\ }\href@noop {} {\bibfield
  {journal} {\bibinfo  {journal} {arXiv preprint arXiv:2111.02345}\ } (\bibinfo
  {year} {2021})}\BibitemShut {NoStop}%
\bibitem [{\citenamefont {Carignan-Dugas}(2019)}]{Carignan2019}%
  \BibitemOpen
  \bibfield  {author} {\bibinfo {author} {\bibfnamefont {A.}~\bibnamefont
  {Carignan-Dugas}},\ }\emph {\bibinfo {title} {A walk through quantum noise: a
  study of error signatures and characterization methods}},\ \href@noop {}
  {Ph.D. thesis},\ \bibinfo  {school} {University of Waterloo} (\bibinfo {year}
  {2019})\BibitemShut {NoStop}%
\bibitem [{\citenamefont {Wallman}(2015)}]{Wallman2015bounding}%
  \BibitemOpen
  \bibfield  {author} {\bibinfo {author} {\bibfnamefont {J.~J.}\ \bibnamefont
  {Wallman}},\ }\href@noop {} {\bibfield  {journal} {\bibinfo  {journal} {arXiv
  preprint arXiv:1511.00727}\ } (\bibinfo {year} {2015})}\BibitemShut {NoStop}%
\bibitem [{\citenamefont {Watrous}(2018)}]{Watrous2018}%
  \BibitemOpen
  \bibfield  {author} {\bibinfo {author} {\bibfnamefont {J.}~\bibnamefont
  {Watrous}},\ }\href@noop {} {\emph {\bibinfo {title} {The theory of quantum
  information}}}\ (\bibinfo  {publisher} {Cambridge University Press},\
  \bibinfo {year} {2018})\BibitemShut {NoStop}%
\bibitem [{\citenamefont {Kitaev}(1997)}]{Kitaev1997}%
  \BibitemOpen
  \bibfield  {author} {\bibinfo {author} {\bibfnamefont {A.~Y.}\ \bibnamefont
  {Kitaev}},\ }\href@noop {} {\bibfield  {journal} {\bibinfo  {journal}
  {Uspekhi Matematicheskikh Nauk}\ }\textbf {\bibinfo {volume} {52}},\ \bibinfo
  {pages} {53} (\bibinfo {year} {1997})}\BibitemShut {NoStop}%
\bibitem [{\citenamefont {Watrous}(2009)}]{Watrous2009}%
  \BibitemOpen
  \bibfield  {author} {\bibinfo {author} {\bibfnamefont {J.}~\bibnamefont
  {Watrous}},\ }\href@noop {} {\bibfield  {journal} {\bibinfo  {journal}
  {Theory Of Computing}\ }\textbf {\bibinfo {volume} {5}},\ \bibinfo {pages}
  {217} (\bibinfo {year} {2009})}\BibitemShut {NoStop}%
\bibitem [{\citenamefont {Emerson}\ \emph {et~al.}(2005)\citenamefont
  {Emerson}, \citenamefont {Alicki},\ and\ \citenamefont
  {{\.Z}yczkowski}}]{Emerson2005}%
  \BibitemOpen
  \bibfield  {author} {\bibinfo {author} {\bibfnamefont {J.}~\bibnamefont
  {Emerson}}, \bibinfo {author} {\bibfnamefont {R.}~\bibnamefont {Alicki}},\
  and\ \bibinfo {author} {\bibfnamefont {K.}~\bibnamefont {{\.Z}yczkowski}},\
  }\href@noop {} {\bibfield  {journal} {\bibinfo  {journal} {Journal of Optics
  B: Quantum and Semiclassical Optics}\ }\textbf {\bibinfo {volume} {7}},\
  \bibinfo {pages} {S347} (\bibinfo {year} {2005})}\BibitemShut {NoStop}%
\bibitem [{\citenamefont {Wallman}(2018)}]{Wallman2018}%
  \BibitemOpen
  \bibfield  {author} {\bibinfo {author} {\bibfnamefont {J.~J.}\ \bibnamefont
  {Wallman}},\ }\href@noop {} {\bibfield  {journal} {\bibinfo  {journal}
  {Quantum}\ }\textbf {\bibinfo {volume} {2}},\ \bibinfo {pages} {47} (\bibinfo
  {year} {2018})}\BibitemShut {NoStop}%
\bibitem [{\citenamefont {Helsen}\ \emph {et~al.}(2020)\citenamefont {Helsen},
  \citenamefont {Roth}, \citenamefont {Onorati}, \citenamefont {Werner},\ and\
  \citenamefont {Eisert}}]{Helsen2020}%
  \BibitemOpen
  \bibfield  {author} {\bibinfo {author} {\bibfnamefont {J.}~\bibnamefont
  {Helsen}}, \bibinfo {author} {\bibfnamefont {I.}~\bibnamefont {Roth}},
  \bibinfo {author} {\bibfnamefont {E.}~\bibnamefont {Onorati}}, \bibinfo
  {author} {\bibfnamefont {A.~H.}\ \bibnamefont {Werner}},\ and\ \bibinfo
  {author} {\bibfnamefont {J.}~\bibnamefont {Eisert}},\ }\href@noop {}
  {\bibfield  {journal} {\bibinfo  {journal} {arXiv preprint arXiv:2010.07974}\
  } (\bibinfo {year} {2020})}\BibitemShut {NoStop}%
\bibitem [{\citenamefont {DiVincenzo}\ \emph {et~al.}(2002)\citenamefont
  {DiVincenzo}, \citenamefont {Leung},\ and\ \citenamefont
  {Terhal}}]{DiVincenzo2002}%
  \BibitemOpen
  \bibfield  {author} {\bibinfo {author} {\bibfnamefont {D.}~\bibnamefont
  {DiVincenzo}}, \bibinfo {author} {\bibfnamefont {D.}~\bibnamefont {Leung}},\
  and\ \bibinfo {author} {\bibfnamefont {B.}~\bibnamefont {Terhal}},\
  }\href@noop {} {\bibfield  {journal} {\bibinfo  {journal} {IEEE Transactions
  on Information Theory}\ }\textbf {\bibinfo {volume} {48}},\ \bibinfo {pages}
  {580} (\bibinfo {year} {2002})}\BibitemShut {NoStop}%
\bibitem [{\citenamefont {Dankert}(2005)}]{Dankert2005}%
  \BibitemOpen
  \bibfield  {author} {\bibinfo {author} {\bibfnamefont {C.}~\bibnamefont
  {Dankert}},\ }\emph {\bibinfo {title} {Efficient simulation of random quantum
  states and operators}},\ \href@noop {} {Master's thesis},\ \bibinfo  {school}
  {University of Waterloo}, \bibinfo {address} {Waterloo, Canada} (\bibinfo
  {year} {2005})\BibitemShut {NoStop}%
\bibitem [{\citenamefont {Dankert}\ \emph {et~al.}(2009)\citenamefont
  {Dankert}, \citenamefont {Cleve}, \citenamefont {Emerson},\ and\
  \citenamefont {Livine}}]{Dankert2009}%
  \BibitemOpen
  \bibfield  {author} {\bibinfo {author} {\bibfnamefont {C.}~\bibnamefont
  {Dankert}}, \bibinfo {author} {\bibfnamefont {R.}~\bibnamefont {Cleve}},
  \bibinfo {author} {\bibfnamefont {J.}~\bibnamefont {Emerson}},\ and\ \bibinfo
  {author} {\bibfnamefont {E.}~\bibnamefont {Livine}},\ }\href@noop {}
  {\bibfield  {journal} {\bibinfo  {journal} {Physical Review A}\ }\textbf
  {\bibinfo {volume} {80}},\ \bibinfo {pages} {012304} (\bibinfo {year}
  {2009})}\BibitemShut {NoStop}%
\bibitem [{\citenamefont {Cleve}\ \emph {et~al.}(2015)\citenamefont {Cleve},
  \citenamefont {Leung}, \citenamefont {Liu},\ and\ \citenamefont
  {Wang}}]{Cleve2015}%
  \BibitemOpen
  \bibfield  {author} {\bibinfo {author} {\bibfnamefont {R.}~\bibnamefont
  {Cleve}}, \bibinfo {author} {\bibfnamefont {D.}~\bibnamefont {Leung}},
  \bibinfo {author} {\bibfnamefont {L.}~\bibnamefont {Liu}},\ and\ \bibinfo
  {author} {\bibfnamefont {C.}~\bibnamefont {Wang}},\ }\href@noop {} {\bibfield
   {journal} {\bibinfo  {journal} {arXiv preprint arXiv:1501.04592}\ }
  (\bibinfo {year} {2015})}\BibitemShut {NoStop}%
\bibitem [{\citenamefont {Zhu}\ \emph {et~al.}(2016)\citenamefont {Zhu},
  \citenamefont {Kueng}, \citenamefont {Grassl},\ and\ \citenamefont
  {Gross}}]{Zhu2016}%
  \BibitemOpen
  \bibfield  {author} {\bibinfo {author} {\bibfnamefont {H.}~\bibnamefont
  {Zhu}}, \bibinfo {author} {\bibfnamefont {R.}~\bibnamefont {Kueng}}, \bibinfo
  {author} {\bibfnamefont {M.}~\bibnamefont {Grassl}},\ and\ \bibinfo {author}
  {\bibfnamefont {D.}~\bibnamefont {Gross}},\ }\href@noop {} {\bibfield
  {journal} {\bibinfo  {journal} {arXiv preprint arXiv:1609.08172}\ } (\bibinfo
  {year} {2016})}\BibitemShut {NoStop}%
\bibitem [{\citenamefont {Webb}(2016)}]{Webb2016}%
  \BibitemOpen
  \bibfield  {author} {\bibinfo {author} {\bibfnamefont {Z.}~\bibnamefont
  {Webb}},\ }\href@noop {} {\bibfield  {journal} {\bibinfo  {journal} {Quantum
  Information and Computation}\ }\textbf {\bibinfo {volume} {16}},\ \bibinfo
  {pages} {1379} (\bibinfo {year} {2016})}\BibitemShut {NoStop}%
\bibitem [{\citenamefont {Zhu}(2017)}]{Zhu2017}%
  \BibitemOpen
  \bibfield  {author} {\bibinfo {author} {\bibfnamefont {H.}~\bibnamefont
  {Zhu}},\ }\href@noop {} {\bibfield  {journal} {\bibinfo  {journal} {Physical
  Review A}\ }\textbf {\bibinfo {volume} {96}},\ \bibinfo {pages} {062336}
  (\bibinfo {year} {2017})}\BibitemShut {NoStop}%
\bibitem [{\citenamefont {Haferkamp}\ \emph {et~al.}(2020)\citenamefont
  {Haferkamp}, \citenamefont {Montealegre-Mora}, \citenamefont {Heinrich},
  \citenamefont {Eisert}, \citenamefont {Gross},\ and\ \citenamefont
  {Roth}}]{Haferkamp2020}%
  \BibitemOpen
  \bibfield  {author} {\bibinfo {author} {\bibfnamefont {J.}~\bibnamefont
  {Haferkamp}}, \bibinfo {author} {\bibfnamefont {F.}~\bibnamefont
  {Montealegre-Mora}}, \bibinfo {author} {\bibfnamefont {M.}~\bibnamefont
  {Heinrich}}, \bibinfo {author} {\bibfnamefont {J.}~\bibnamefont {Eisert}},
  \bibinfo {author} {\bibfnamefont {D.}~\bibnamefont {Gross}},\ and\ \bibinfo
  {author} {\bibfnamefont {I.}~\bibnamefont {Roth}},\ }\href@noop {} {\bibfield
   {journal} {\bibinfo  {journal} {arXiv preprint arXiv:2002.09524}\ }
  (\bibinfo {year} {2020})}\BibitemShut {NoStop}%
\bibitem [{\citenamefont {Skanes-Norman}(2021)}]{Skanes2021}%
  \BibitemOpen
  \bibfield  {author} {\bibinfo {author} {\bibfnamefont {J.}~\bibnamefont
  {Skanes-Norman}},\ }\emph {\bibinfo {title} {On Realistic Errors in Quantum
  Computers}},\ \href@noop {} {Master's thesis},\ \bibinfo  {school}
  {University of Waterloo} (\bibinfo {year} {2021})\BibitemShut {NoStop}%
\bibitem [{\citenamefont {Graydon}\ \emph {et~al.}(2021)\citenamefont
  {Graydon}, \citenamefont {Skanes-Norman},\ and\ \citenamefont
  {Wallman}}]{Graydon2021}%
  \BibitemOpen
  \bibfield  {author} {\bibinfo {author} {\bibfnamefont {M.~A.}\ \bibnamefont
  {Graydon}}, \bibinfo {author} {\bibfnamefont {J.}~\bibnamefont
  {Skanes-Norman}},\ and\ \bibinfo {author} {\bibfnamefont {J.~J.}\
  \bibnamefont {Wallman}},\ }\href@noop {} {\bibfield  {journal} {\bibinfo
  {journal} {arXiv preprint arXiv:2108.04200}\ } (\bibinfo {year}
  {2021})}\BibitemShut {NoStop}%
\bibitem [{\citenamefont {Roy}\ and\ \citenamefont {Scott}(2009)}]{Roy2009}%
  \BibitemOpen
  \bibfield  {author} {\bibinfo {author} {\bibfnamefont {A.}~\bibnamefont
  {Roy}}\ and\ \bibinfo {author} {\bibfnamefont {A.~J.}\ \bibnamefont
  {Scott}},\ }\href@noop {} {\bibfield  {journal} {\bibinfo  {journal}
  {Designs, codes and cryptography}\ }\textbf {\bibinfo {volume} {53}},\
  \bibinfo {pages} {13} (\bibinfo {year} {2009})}\BibitemShut {NoStop}%
\bibitem [{\citenamefont {Lee}\ and\ \citenamefont {Watrous}(2020)}]{Lee2020}%
  \BibitemOpen
  \bibfield  {author} {\bibinfo {author} {\bibfnamefont {C.~D.-Y.}\
  \bibnamefont {Lee}}\ and\ \bibinfo {author} {\bibfnamefont {J.}~\bibnamefont
  {Watrous}},\ }\href@noop {} {\bibfield  {journal} {\bibinfo  {journal}
  {Quantum}\ }\textbf {\bibinfo {volume} {4}},\ \bibinfo {pages} {253}
  (\bibinfo {year} {2020})}\BibitemShut {NoStop}%
\bibitem [{\citenamefont {Fuchs}\ and\ \citenamefont {Van
  De~Graaf}(1999)}]{Fuchs1999}%
  \BibitemOpen
  \bibfield  {author} {\bibinfo {author} {\bibfnamefont {C.~A.}\ \bibnamefont
  {Fuchs}}\ and\ \bibinfo {author} {\bibfnamefont {J.}~\bibnamefont {Van
  De~Graaf}},\ }\href@noop {} {\bibfield  {journal} {\bibinfo  {journal} {IEEE
  Transactions on Information Theory}\ }\textbf {\bibinfo {volume} {45}},\
  \bibinfo {pages} {1216} (\bibinfo {year} {1999})}\BibitemShut {NoStop}%
\bibitem [{\citenamefont {Magesan}\ \emph {et~al.}(2012)\citenamefont
  {Magesan}, \citenamefont {Gambetta},\ and\ \citenamefont
  {Emerson}}]{Magesan2012}%
  \BibitemOpen
  \bibfield  {author} {\bibinfo {author} {\bibfnamefont {E.}~\bibnamefont
  {Magesan}}, \bibinfo {author} {\bibfnamefont {J.~M.}\ \bibnamefont
  {Gambetta}},\ and\ \bibinfo {author} {\bibfnamefont {J.}~\bibnamefont
  {Emerson}},\ }\href@noop {} {\bibfield  {journal} {\bibinfo  {journal}
  {Physical Review A}\ }\textbf {\bibinfo {volume} {85}},\ \bibinfo {pages}
  {042311} (\bibinfo {year} {2012})}\BibitemShut {NoStop}%
\bibitem [{\citenamefont {Wallman}\ and\ \citenamefont
  {Flammia}(2014)}]{Wallman2014}%
  \BibitemOpen
  \bibfield  {author} {\bibinfo {author} {\bibfnamefont {J.~J.}\ \bibnamefont
  {Wallman}}\ and\ \bibinfo {author} {\bibfnamefont {S.~T.}\ \bibnamefont
  {Flammia}},\ }\href@noop {} {\bibfield  {journal} {\bibinfo  {journal} {New
  Journal of Physics}\ }\textbf {\bibinfo {volume} {16}},\ \bibinfo {pages}
  {103032} (\bibinfo {year} {2014})}\BibitemShut {NoStop}%
\bibitem [{\citenamefont {Halmos}(2013)}]{halmos2013measure}%
  \BibitemOpen
  \bibfield  {author} {\bibinfo {author} {\bibfnamefont {P.~R.}\ \bibnamefont
  {Halmos}},\ }\href@noop {} {\emph {\bibinfo {title} {Measure theory}}},\
  Vol.~\bibinfo {volume} {18}\ (\bibinfo  {publisher} {Springer},\ \bibinfo
  {year} {2013})\BibitemShut {NoStop}%
\bibitem [{\citenamefont {Knapp}(2016)}]{Knapp2016}%
  \BibitemOpen
  \bibfield  {author} {\bibinfo {author} {\bibfnamefont {A.~W.}\ \bibnamefont
  {Knapp}},\ }\href@noop {} {\emph {\bibinfo {title} {Representation theory of
  semisimple groups}}}\ (\bibinfo  {publisher} {Princeton university press},\
  \bibinfo {year} {2016})\BibitemShut {NoStop}%
\bibitem [{\citenamefont {Kirillov}(1976)}]{Kirillov1976}%
  \BibitemOpen
  \bibfield  {author} {\bibinfo {author} {\bibfnamefont {A.}~\bibnamefont
  {Kirillov}},\ }\href@noop {} {\emph {\bibinfo {title} {Elements of the theory
  of representations}}}\ (\bibinfo  {publisher} {Springer},\ \bibinfo {year}
  {1976})\BibitemShut {NoStop}%
\bibitem [{\citenamefont {Bannai}\ \emph {et~al.}(2019)\citenamefont {Bannai},
  \citenamefont {Nakahara}, \citenamefont {Zhao},\ and\ \citenamefont
  {Zhu}}]{Bannai2019}%
  \BibitemOpen
  \bibfield  {author} {\bibinfo {author} {\bibfnamefont {E.}~\bibnamefont
  {Bannai}}, \bibinfo {author} {\bibfnamefont {M.}~\bibnamefont {Nakahara}},
  \bibinfo {author} {\bibfnamefont {D.}~\bibnamefont {Zhao}},\ and\ \bibinfo
  {author} {\bibfnamefont {Y.}~\bibnamefont {Zhu}},\ }\href@noop {} {\bibfield
  {journal} {\bibinfo  {journal} {Journal of Physics A: Mathematical and
  Theoretical}\ }\textbf {\bibinfo {volume} {52}},\ \bibinfo {pages} {495301}
  (\bibinfo {year} {2019})}\BibitemShut {NoStop}%
\bibitem [{\citenamefont {Group}(2021)}]{GAP4}%
  \BibitemOpen
  \bibfield  {author} {\bibinfo {author} {\bibfnamefont {T.~G.}\ \bibnamefont
  {Group}},\ }\href@noop {} {\bibinfo {title} {Gap -- groups, algorithms, and
  programming, version 4.11.1}} (\bibinfo {year} {2021})\BibitemShut {NoStop}%
\bibitem [{\citenamefont {Appleby}(2005)}]{Appleby2005}%
  \BibitemOpen
  \bibfield  {author} {\bibinfo {author} {\bibfnamefont {D.~M.}\ \bibnamefont
  {Appleby}},\ }\href@noop {} {\bibfield  {journal} {\bibinfo  {journal}
  {Journal of Mathematical Physics}\ }\textbf {\bibinfo {volume} {46}}
  (\bibinfo {year} {2005})}\BibitemShut {NoStop}%
\bibitem [{\citenamefont {Erhard}\ \emph {et~al.}(2019)\citenamefont {Erhard},
  \citenamefont {Wallman}, \citenamefont {Postler}, \citenamefont {Meth},
  \citenamefont {Stricker}, \citenamefont {Martinez}, \citenamefont
  {Schindler}, \citenamefont {Monz}, \citenamefont {Emerson},\ and\
  \citenamefont {Blatt}}]{Erhard2019}%
  \BibitemOpen
  \bibfield  {author} {\bibinfo {author} {\bibfnamefont {A.}~\bibnamefont
  {Erhard}}, \bibinfo {author} {\bibfnamefont {J.~J.}\ \bibnamefont {Wallman}},
  \bibinfo {author} {\bibfnamefont {L.}~\bibnamefont {Postler}}, \bibinfo
  {author} {\bibfnamefont {M.}~\bibnamefont {Meth}}, \bibinfo {author}
  {\bibfnamefont {R.}~\bibnamefont {Stricker}}, \bibinfo {author}
  {\bibfnamefont {E.~A.}\ \bibnamefont {Martinez}}, \bibinfo {author}
  {\bibfnamefont {P.}~\bibnamefont {Schindler}}, \bibinfo {author}
  {\bibfnamefont {T.}~\bibnamefont {Monz}}, \bibinfo {author} {\bibfnamefont
  {J.}~\bibnamefont {Emerson}},\ and\ \bibinfo {author} {\bibfnamefont
  {R.}~\bibnamefont {Blatt}},\ }\href@noop {} {\bibfield  {journal} {\bibinfo
  {journal} {Nature communications}\ }\textbf {\bibinfo {volume} {10}},\
  \bibinfo {pages} {1} (\bibinfo {year} {2019})}\BibitemShut {NoStop}%
\bibitem [{\citenamefont {Wallman}\ and\ \citenamefont
  {Emerson}(2016)}]{Wallman2016}%
  \BibitemOpen
  \bibfield  {author} {\bibinfo {author} {\bibfnamefont {J.~J.}\ \bibnamefont
  {Wallman}}\ and\ \bibinfo {author} {\bibfnamefont {J.}~\bibnamefont
  {Emerson}},\ }\href@noop {} {\bibfield  {journal} {\bibinfo  {journal}
  {Physical Review A}\ }\textbf {\bibinfo {volume} {94}},\ \bibinfo {pages}
  {052325} (\bibinfo {year} {2016})}\BibitemShut {NoStop}%
\bibitem [{\citenamefont {Bengtsson}\ and\ \citenamefont
  {{\.Z}yczkowski}(2017)}]{Bengtsson2017}%
  \BibitemOpen
  \bibfield  {author} {\bibinfo {author} {\bibfnamefont {I.}~\bibnamefont
  {Bengtsson}}\ and\ \bibinfo {author} {\bibfnamefont {K.}~\bibnamefont
  {{\.Z}yczkowski}},\ }\href@noop {} {\bibfield  {journal} {\bibinfo  {journal}
  {arXiv preprint arXiv:1701.07902}\ } (\bibinfo {year} {2017})}\BibitemShut
  {NoStop}%
\end{thebibliography}%

\end{document}